\documentclass[11pt,reqno]{amsart}
\usepackage[foot]{amsaddr} 

\usepackage{amsmath}
\usepackage{amssymb}
\usepackage{graphicx}
\usepackage{xcolor}
\usepackage{hyperref}
\hypersetup{
	colorlinks,
	linkcolor={black!30!blue},
	citecolor={black!30!blue},
	urlcolor={black!30!blue},
	breaklinks={true}
}

\usepackage{bbm}
\usepackage{cite}
\usepackage{array}

\usepackage[margin=2.5cm]{geometry}

\makeatletter
\def\subsection{\@startsection{subsection}{2}%
	\z@{.5\linespacing\@plus.7\linespacing}{.5\linespacing}%
	{\normalfont\scshape\centering}}
\def\subsubsection{\@startsection{subsubsection}{2}%
	\z@{.5\linespacing\@plus.7\linespacing}{.5\linespacing}%
	{\normalfont\scshape\centering}}
\makeatother

\numberwithin{equation}{section}
\newtheorem{thm}{Theorem}[section]

\newtheorem{lem}[thm]{Lemma}
\newtheorem{prop}[thm]{Proposition}

\newtheorem{assume}{Assumption}

\def\ititem#1#2{\par\vskip12pt\noindent\ifx#1.$\bullet$\else(#1)\fi\ {\it#2}\medskip\\}

\def\tr{\operatorname{tr}}
\def\idty{\mathbbm{1}}

\def\Rl{{\mathbb R}}\def\Cx{{\mathbb C}}
\def\Ir{{\mathbb Z}}\def\Nl{{\mathbb N}}

\def\norm #1{\Vert #1\Vert}

\def\dom{{\mathop{\rm dom}\nolimits}\,}
\def\bra #1{\langle #1\vert}
\def\ket #1{\vert #1\rangle}
\def\braket #1#2{\langle #1 \mid #2\rangle}
\def\ketbra #1#2{\vert #1\rangle \langle #2\vert}
\def\kettbra#1{\ketbra{#1}{#1}}

\def\tr{\mathop{\rm tr}\nolimits}
\def\abs#1{\vert#1\vert}

\def\Ell{{\mathcal L}}

\def\Order{{\bf O}}

\def\spec{{\rm spec}}

\def\cdt{{\cdot}} 

\def\im{\Im m}
\def\co{\mathrm{conv}}

\def\Fou{{\mathcal F}} 
\def\Brill{{\mathbb B}} 

\def\HH{{\mathcal H}}\def\KK{{\mathcal K}}

\def\Ell{{\mathcal L}}
\def\DD{{\mathcal D}}

\def\expi#1{e^{i#1}}
\def\ratef{I} 
\def\ratefr{\widetilde\ratef}
\def\rateL{R} 
\def\rateLr{\widetilde\rateL}
\def\level{\Lambda}
\def\spr{\mathop{\rm spr}}

\def\arsinh{\mathop{\rm arcsinh}}

\begin{document}
	
\title{Exponential tail estimates for quantum lattice dynamics}

\author[C. Cedzich]{Christopher Cedzich${}^1$}
\email{\href{mailto:cedzich@hhu.de}{cedzich@hhu.de}}
\address{${}^1$Quantum Technology Group, Heinrich Heine Universit\"at D\"usseldorf, Universit\"atsstr. 1, 40225 D\"usseldorf, Germany}

\author[A. Joye]{Alain Joye${}^2$}
\email{\href{mailto:alain.joye@univ-grenoble-alpes.fr}{alain.joye@univ-grenoble-alpes.fr}}
\address{${}^2$Universit\'e Grenoble Alpes, CNRS Institut Fourier, 38000 Grenoble, France}

\author[A.H. Werner]{Albert H. Werner${}^3$}
\email{\href{mailto:werner@math.ku.dk}{werner@math.ku.dk}}
\address{${}^3${QMATH}, Department of Mathematical Sciences, University of Copenhagen, Universitetsparken 5, 2100 Copenhagen, Denmark}

\author[R.F. Werner]{Reinhard F. Werner${}^4$}
\email{\href{mailto:reinhard.werner@itp.uni-hannover.de}{reinhard.werner@itp.uni-hannover.de}}
\address{${}^4$Institut für Theoretische Physik, Leibniz Universität Hannover, Appelstr. 2, 30167 Hannover, Germany}

\begin{abstract}
	We consider the quantum dynamics of a particle on a lattice for large times. Assuming translation invariance, and either discrete or continuous time parameter, the distribution of the ballistically scaled position $Q(t)/t$ converges weakly to a distribution that is compactly supported in velocity space, essentially the distribution of group velocity in the initial state. We show that the total probability of velocities strictly outside the support of the asymptotic measure goes to zero exponentially with $t$, and we provide a simple method to estimate the exponential rate uniformly in the initial state. Near the boundary of the allowed region the rate function goes to zero like the power $3/2$ of the distance to the boundary.
	The method is illustrated in several examples.
\end{abstract}

\maketitle  

\section{Introduction}
It is an almost universally accepted view that at very small scales the model of spacetime as a continuum must be modified. Yet, one of the most natural ways of implementing this idea, namely replacing the continuum by a lattice, has been plagued by the somewhat arbitrary breaking of isotropy involved in choosing the lattice axes. One thus needs to explain how at large scales an exactly circular light cone can govern the propagation in a  cubic lattice, a case of inverse symmetry breaking \cite{weinbergGaugeGlobalSymmetries1974}. Moreover, one would like to have an explicit bound on the fuzziness of the cone: The probability for ``faster-than-light'' signals will be non-zero in such a model, but these signals should be weak on a large scale, and decrease rapidly for speeds appreciably larger than the speed of light of the continuum model. In this paper we will provide a technique to get such bounds, namely large deviation type \cite{Ellis} estimates: The total probability for a quantum particle to be found outside the cone, say at any speed larger by $\varepsilon$ than allowed, decreases exponentially in time which is sometimes called ``evanescence'' of the particle's wave function in analogy to optical waves \cite{berryEvanescentRealWaves1994}. We provide a formula for the rate of exponential decrease as a function of $\varepsilon$. Typically, the total probability outside an $\varepsilon$-neighbourhood of the cone goes like $\exp(-c\,\varepsilon^{3/2}t)$.

Effective propagation cones have, of course, been noted before. In the context of quantum walks (discrete time unitary lattice dynamics) they appear in the  asymptotic form of the position $Q(t)$ at large times $t$, scaled as a velocity by considering $Q(t)/t$ \cite{Scudo,timerandom,spacetimerandom}. As in the solid state context (continuous time Hamiltonian lattice dynamics), this scaled asymptotic position converges to the group velocity operator. The spectrum of this operator thus defines the propagation region, which is not necessarily a spherical cone and might not even be convex. Spherical propagation regions are often related to the appearance of ``Dirac points'', certain singular points in the band structure of lattice systems such as graphene. However, the mere convergence of $Q(t)/t$ only allows for the conclusion that the total probability outside the cone goes to zero, not for how fast and not for how the rate depends on the distance from the cone.

Our technique for providing such bounds on the rate of convergence applies to both discrete time and continuous time Hamiltonian dynamics. We assume only translation invariance and that the number of states per unit cell is finite. It is closely related to stationary phase/steepest descent evaluation \cite{debyeNaeherungsformelnFuerZylinderfunktionen1909}. However, it seems more direct, especially if one does not want to estimate the probability amplitudes for individual points, but the total probability for large regions outside the cone.

Our paper is organized as follows: in the next section we describe the setting in which we are working and state the main estimate which is proved in Section \ref{sec:main_proof}. In Section \ref{sec:discussion} we discuss this estimate and the limit cases, and show several properties of the involved quantities. Emblematic examples are discussed in Section \ref{sec:Exmp}.
\section{Main result}
\subsection{Setting and statement of the main estimate}
We consider systems, whose configurational variable (for a single particle: the position) is confined to a point lattice $X$ in $s$-dimensional Euclidean space. At each lattice site $x$ we allow some finite dimensional Hilbert space $\KK_x\cong\KK=\Cx^d$ of internal states with $d<\infty$. The overall system Hilbert space is thus $\HH=\ell^2(\Ir^s)\otimes\KK$. We typically write $\psi\in\HH$ as a $\KK$-valued square summable function on $\Ir^s$ so that the $\alpha$-component of the position operator is the multiplication operator $(Q_\alpha\psi)(x)=x_\alpha\psi(x)$.

In such a system we consider a coherent quantum mechanical time evolution, either generated by a Hamiltonian $H$ (the continuous time case) or by a unitary operator $W$ in $\HH$ (the discrete time case). By $W^t$ we denote the evolution after time $t$, i.e., $W^t=\exp(-iHt)$ in the continuous time case and the (integer) matrix power in the discrete time case. In either case we require that $W^t$ commutes with the lattice translations. Thus the dynamics is partially diagonalized by the Fourier transform
\begin{equation}\label{Fourier}
  (\Fou\psi)(p)=\sum_xe^{-ip\cdot x}\ \psi(x)=:\hat\psi(p),
\end{equation}
where $p\in[-\pi,\pi]^s=\Brill$, the Brillouin zone, which is to be considered with periodic boundary conditions, i.e., as the $s$-torus. We equip $\Brill$ with the normalized Haar measure $d^sp/(2\pi)^s$, which makes $\Fou$ unitary from $\HH$ to $\widehat\HH=\Ell^2(\Brill)\otimes\KK$. By translation invariance, $W$ or $H$ become matrix valued multiplication operators, i.e.,
\begin{equation}\label{Wp}
  (\Fou W\Fou^*\ \hat\psi)(p)=W(p)\hat\psi(p),
\end{equation}
and similarly for $H$. Since $\hat\psi$ is $\KK=\Cx^d$-valued, $W(p)$ is a unitary $d\times d$-matrix for every $p\in\Brill$. For example, the ``Hadamard walk'', a standard example of a one-dimensional quantum walk with $\KK=\Cx^2$, has
\begin{equation}\label{hadamard}
  W(p)=\frac1{\sqrt2}\begin{pmatrix}
         e^{ip} & e^{ip}  \\-e^{-ip} & e^{-ip} \end{pmatrix}.
\end{equation}
This is a nearest neighbour dynamics, because only $e^{\pm ip}$ and no higher powers appear. Generally, a walk has maximal jump length $L$ if each entry of $W(p)$ is a trigonometric polynomial of absolute degree $\leq L$. Similar remarks apply to $H(p)$.
We generalize this further by allowing $W(p)$ to be an infinite series. Our standing ``locality'' assumption is then:

\begin{assume}\label{assumption:analytic_ext}
	$W(p)$ has an entire analytic extension from $\Rl^s$ to $\Cx^s$.
\end{assume}

Note that the analogous assumption on $H(p)$ carries over to $W^t=\exp(-itH)$. The growth of $W$ in the imaginary direction is known to express decay properties of the evolution. For example, when $W$ has finite jump length $L$, $W(p)$ is of exponential type, i.e., satisfies a bound $\norm{W^t(p+i\lambda)}\leq c\,\exp(tL\abs\lambda)$, which by the Paley-Wiener theorem implies that an initially localized state $\rho$ is strictly localized in a ball of radius $tL$. Of course, this is anyhow obvious from the meaning of the jump length, but it shows how the growth of the analytic extension relates to propagation. Our main result will refine this connection.

We are interested in the position distribution for large times. Since there is no dissipation, particles will generically spread ballistically, i.e., the scaled quantity $Q(t)/t=(1/t)(W^t)^*QW^t$ is expected to have a limit distribution. For any measurable set $M\subset\Rl^s$ we introduce the probability
\begin{equation}\label{pmdef}
    p_t(\rho,M)=\tr\bigl(\rho\,\chi_M(Q(t)/t)\bigr)
\end{equation}
for finding a value $q\in tM$ in a position measurement at time $t$, starting from the initial state $\rho$. Here $\chi_M$ denotes the indicator function of the set $M$, so $\chi_M(Q)$ is the operator of multiplication by $\chi_M(x_1,\ldots,x_s)$. It turns out that in the described setting the limit $t\to\infty$ of \eqref{pmdef} exists in the weak sense, and all the limit measures have their support in a compact, $\rho$-independent set $\Gamma\subset\Rl^s$, which we call the {\it propagation region} in velocity space \cite{Scudo,timerandom}.  A precise formulation will also be given in Sect.~\ref{sec:propregion}, together with an explicit formula for the limits and proofs.

Our main interest here is the behaviour of the probability measures $p_t$ outside of the propagation region. From the statements given, it is clear that when $M\cap\Gamma=0$, we get $p_t(\rho,M)\to0$ for all $\rho$. But how fast is this convergence? Our main theorem will establish that it is exponential, with a rate increasing with the distance between $\Gamma$ and $M$, i.e.,
\begin{equation}\label{expBound}
    p_t(\rho,M)\leq c(M)\, e^{-t\,\ratef(M)},
\end{equation}
where $\ratef(M)$ is an exponential rate that is independent of the initial state. If $M$ is split into different regions $M=M_1\cup M_2$, estimates of this type are dominated by the smaller rate, i.e., when $\ratef(M_1)<\ratef(M_2)$, we have $p_t(\rho,M_1)\gg p_t(\rho,M_2)$ for large $t$, and hence $\ratef(M)=\ratef(M_1)$. In this sense only the points $x\in M$ with the lowest rate ``$\ratef(x)$'' contribute. We say that a lower semi-continuous function $\ratef:\Rl^n\to\Rl_+\cup\{+\infty\}$ is a {\it rate function} for the process, if
\begin{equation}\label{LDbound}
    \limsup_{t\to\infty} \frac1t\, \log p_t(\rho,M)\leq -\ratef(M)\equiv-\inf_{x\in M}{\ratef(x)}
\end{equation}
for every $\rho$ which initially has compact support, and every closed set $M$. 
We will work with the form \eqref{LDbound}, and regard \eqref{expBound} as a heuristic interpretation. Strictly speaking, this allows a sub-exponential time dependence of the constant, or else \eqref{expBound} is to be read as a bound $\leq c_\varepsilon(M)\exp(-t(I(M)-\varepsilon))$ for every $\varepsilon>0$ with a time-independent constant $c_\varepsilon(M)$.

Estimates of this type belong to the theory of Large Deviations \cite{Ellis}. Typically, in that theory one also shows lower bounds, but since these would be state dependent, we do not consider them here.

With these preparations our estimate can be stated as follows:

\begin{thm}\label{mainprop}
In the setting described above, a rate function $\ratef$ in \eqref{LDbound} can be taken as the Legendre transform
\begin{equation}\label{Legendre}
    \ratef(x)=\sup_\lambda\bigl\{\lambda\cdot x-\rateL(\lambda)\bigr\}
\end{equation}
of a function $\rateL: {\mathbb R}^s \mapsto {\mathbb R}$ that can be computed directly from $H(p)$, respectively $W(p)$, as
\begin{align}\label{formula}
  \rateL(\lambda)&= \sup_{p\in\Brill}\sup\Bigl\{2\im(\omega) \Bigm| \omega  \in\spec\Bigl(H\bigr(p+\tfrac{i\lambda}2\bigr)\Bigr)\Bigr\} , \intertext{respectively}
  \rateL(\lambda)&= \sup_{p\in\Brill}\sup\Bigl\{\log\abs u^2 \Bigm| u \in\spec\Bigl(W\bigr(p+\tfrac{i\lambda}2\bigr)\Bigr)\Bigr\},
  \label{formulaW}
\end{align}
where $\Brill=[-\pi,\pi]^s$ denotes the Brillouin zone.
\end{thm}

\subsection{Group velocity and propagation region}\label{sec:propregion}
Here we provide short proofs of the general claims made about the limiting velocity distribution. These are known results \cite{Scudo,timerandom}, so our main reason for presenting them here is to introduce the proof method, a variant of which will be central to the proof of Theorem \ref{mainprop}. We begin by diagonalizing $W(p)$ or $H(p)$ for every $p\in\Brill$:
\begin{equation}\label{spectral}
	\begin{aligned}
	  H(p)&=\sum_j\omega_j(p)P_j(p) \\
	  W(p)&=\sum_je^{-i\omega_j(p)}\,P_j(p),
	\end{aligned}
\end{equation}
where the $\omega_j$ determine the eigenvalues ($\omega_j$ is only defined $\bmod2\pi$ in the discrete time case) and the $P_j$ are rank-1 eigenprojections. Obviously, we have to be careful about degeneracies. Unlike in the one-dimensional ($s=1$) case, the analyticity of $W$ does not imply that we can also choose the branch functions $\omega_j$ analytic, not even locally. On the other hand, analyticity near $p$ is easy to see for any nondegenerate branch $\omega_j(p)$: In that case we can set up a Cauchy Resolvent Integral formula around the isolated point $\omega_j(p)$ in the complex plane providing analytic expressions for both $P_j$ and $\omega_j$. The same holds for isolated branches with higher multiplicity, which arise by tensoring with an additional internal space $\KK'$, on which $W$ acts like the identity. Points $(p,\omega_i(p))$ for which such an analytic choice is possible are called {\it regular} in the theory of analytic matrix functions \cite[\S~S3]{Baumgrtl}. It is shown there that the regular points form a connected open submanifold, so the irregular/singular points are of Lebesgue measure zero. Thus the expressions \eqref{spectral} make sense almost everywhere, including small open neighbourhoods. Even degenerate points may be regular, where two bands just form intersecting analytic manifolds. The prototype of a singular point is the tip of a cone, as for example $H(p_1,p_2)=p_1\sigma_1+p_2\sigma_2$ in $\Cx^2$, where $\sigma_1$ and $\sigma_2$ are Pauli matrices. Even when there are no singular points, the bands may wrap around the torus $\Brill$ in a non-trivial way. That is, following a band $\omega_j(p)$ along a closed, but not contractible path one may end up in another branch.

With this information we can introduce the {\it group velocity operator}. This operator also commutes with translations so that
it acts by multiplication with a $p$-dependent matrix. It is the vector operator $V$ whose $\alpha$-coordinate is defined (at every regular point) as
\begin{equation}\label{groupvel}
    V_\alpha(p)=\sum_j\frac{\partial\omega_j(p)}{\partial p_\alpha}\ P_j(p).
\end{equation}
An alternative definition, which does not require diagonalization, but obviously coincides with the above at every regular point is
\begin{equation}\label{groupvelalcont}
    V_\alpha(p)=\lim_{T\rightarrow \infty }\frac{1}{T}\int_0^T e^{itH(p)}\frac{\partial H(p)}{\partial p_\alpha}e^{-itH(p)}dt= \sum_\mu \widetilde P_\mu(p)\frac{\partial H(p)}{\partial p_\alpha}\widetilde P_\mu(p)
\end{equation}
and, for the discrete time case,
\begin{equation}
	V_\alpha(p)=\lim_{T\rightarrow \infty }\frac{2i}{T(T+1)}\sum_{t=0}^TW(p)^{-t}\frac{\partial W(p)^t}{\partial p_\alpha},
\end{equation}
where the $\widetilde P_\mu$ are now the non-degenerate projections of the functional calculus belonging to the distinct eigenvalues.
The key feature is that the components of $V$ commute with $H$ or $W$ and almost everywhere with each other. That is, as matrix-valued multiplication operators on $\Ell^2(\Brill)\otimes\KK$, the $V_\alpha$ do commute and hence have a joint functional calculus, so we can evaluate bounded measurable functions $f:\Rl^d\to\Rl$ on $V$, getting an operator $f(V)$. The spectral projection for $M\subset\Rl^s$ is then $\chi_M(V)$, and the \emph{joint spectrum} $\Gamma$ is the complement of the largest open set $U\subset\Rl^s$ such that $\chi_U(V)=0$. Since $p\mapsto H(p)$ is assumed to be analytic and $\Brill$ is compact, $V$ is also bounded, and hence $\Gamma$ is compact. We can construct almost-eigenvectors of $V(p)$ for eigenvalue tuple $\nabla\omega_j(p)$ at every regular point, so we have that $\nabla\omega_j(p)\in\Gamma$ at every regular point, and since $V$ acts almost everywhere by multiplication with such numbers we have
\begin{equation}\label{specV}
  \Gamma=\{\nabla\omega_j(p)|\,p\in\Rl^s\ \text{regular},\:j=1,\ldots d\}^{\text{closure}}.
\end{equation}

The reason to study these objects, and the justification for calling $\Gamma$ the propagation region, is given in the following Proposition. It is basically well-known, and the reason for calling $V$ the group velocity operator. In the walk context the first statement (apart from worked individual examples) seems to be \cite{Scudo}. The technique there is based on the convergence of moments, which is slightly problematic, because the moments might fail to exist for the initial state and all through the evolution. We give here a proof based on characteristic functions, which appeared in \cite{timerandom} in a much more complex context, which perhaps obscured the structure of the argument. However, it is this basic idea that will be needed later for the proof of Theorem \ref{mainprop}.

\begin{prop}\label{prop:Velo}
For every density operator $\rho$, the probability measure $p_t(\rho,\cdot)$ defined in \eqref{pmdef} converges weakly to the distribution of $V$ in $\rho$. Explicitly:
\begin{equation}\label{limf}
    \lim_{t\to\infty}\tr\rho f(Q(t)/t)=\tr\rho f(V)
\end{equation}
for every continuous function $f:\Rl^s\to\Rl$ vanishing at infinity, with $f$ evaluated on both sides in the respective functional calculus.
\end{prop}

\begin{proof}
The weak convergence is equivalent by the continuity theorem for characteristic functions \cite{Lukacz} (aka. Levy's convergence theorem \cite{levyDeterminationLoisProbabilite1922}\cite[Théorème 17]{levyTheorieAdditionVariables1937}) to the pointwise convergence of characteristic functions, i.e., to the special case of \eqref{limf} for $f(x)=\exp(ix\cdt\lambda)$ for all $\lambda$. We fix one such $\lambda\in\Rl^s$ now.

Exponentials of $Q$ are shift operators in the momentum representation. To utilize this, we introduce the operators
\begin{equation}\label{Wlambda}
  W_\lambda:=e^{i \lambda\cdot Q}We^{-i \lambda\cdot Q}.
\end{equation}
This is an operator that once again commutes with translations. So in the momentum representation on $\widehat\HH=\Ell^2(\Brill)\otimes\KK$ it acts as a matrix multiplication operator, namely
\begin{equation}\label{Wlambdap}
  W_\lambda(p)=W(p-\lambda).
\end{equation}
With the same transformation for $W^t$ we get
\begin{equation}\label{Wtlam}
  e^{i \lambda\cdot Q}W^t e^{-i \lambda\cdot Q}=(W_\lambda)^t=e^{-iH_\lambda t},
\end{equation}
where the second equality is for the continuous time case, with $W^t=\exp(-iHt)$, and $H_\lambda=\exp(i\lambda\cdot Q)H\exp(-i\lambda\cdot Q)$, which multiplies in momentum representation with $H_\lambda(p)=H(p-\lambda)$. It is worth noting that while the operators $W_\lambda$ on the whole space $\HH$ are unitarily conjugate to each other by virtue of \eqref{Wlambda}, the matrix $W_\lambda(p)$ is not conjugate to $W(p)$ and generally these have different spectrum.

We now claim (in either case) the strong limit formula
\begin{equation}\label{claimV}
    \lim_{t\to\infty}W^{-t}W^t_{\lambda/t}=e^{i\lambda\cdot V}.
\end{equation}
For the proof of this claim note first that all operators involved are unitary, so strong convergence is equivalent to weak convergence, i.e., we only have to show the convergence of all matrix elements.  Moreover, all operators in \eqref{claimV} commute with translations, and hence act by multiplication with uniformly bounded matrices in momentum space. By the Dominated Convergence Theorem it therefore suffices to show that
\begin{equation}\label{domcond}
  W^{-t}(p)\,W^t_{\lambda/t}(p)\longrightarrow e^{i\lambda\cdot V(p)}\quad  \text{for almost all}\ p.
\end{equation}
We may therefore assume that $p$ is a regular point, and that in the neighbourhood of $p$ the eigenprojections and eigenvalues are analytic functions. By \eqref{spectral} the left hand side is hence
\begin{equation}\label{prodwwl}
  W(p)^{-t}W_{\lambda/t}(p)^t=\sum_{j,\ell} e^{-it \bigl(\omega_\ell(p-\lambda/t)-\omega_j(t)\bigr)} P_j(p)P_\ell(p-\lambda/t).
\end{equation}
Each of the finitely many terms with $\ell\neq j$ goes to zero, because $P_\ell(p-\lambda/t)\to P_\ell(p)$ as $t\to\infty$, and these operators are multiplied by a bounded function. For the diagonal terms $\ell= j$ the exponent converges to a directional derivative as $t\to\infty$:
\begin{equation}\label{domega}
 -t \bigl(\omega_j(p-\lambda/t)-\omega_j(p)\bigr)\quad\longrightarrow\quad \sum_\alpha \lambda_\alpha \frac{\partial \omega_j(p)}{\partial p_\alpha}.
\end{equation}
Comparing with \eqref{groupvel}, this is precisely the exponent of the corresponding term in $\exp{i\lambda\cdt V}$ which proves \eqref{claimV}.

Now consider an arbitrary pure state density operator $\rho=\kettbra\psi$. Then the characteristic function of its scaled position at time $t$ is
\begin{equation}\label{VpureFint}
 \bigl\langle\psi\bigm|W^{-t}e^{i\lambda\cdot Q/t}W^t\psi\bigr\rangle
  =\bigl\langle\psi\bigm|\bigl(W^{-t}e^{i\lambda\cdot Q/t}W^te^{-i\lambda\cdot Q/t}\bigr)e^{i\lambda\cdot Q/t}\psi\bigr\rangle
   =\bigl\langle\bigl(W^{-t}_{\lambda/t}W^t\bigr)\psi\bigm|e^{i\lambda\cdot Q/t}\psi\bigr\rangle
\end{equation}
According to the claim just proved the first vector goes in norm to $\exp(-i\lambda\cdot V)\psi$ and the second goes to $\psi$ by strong continuity of the one-parameter group generated by $Q$. This proves the proposition for pure states. For mixed states one can use a trace norm approximation by a linear combination of pure states.
\end{proof}

\section{Proof of Theorem \ref{mainprop}}
\label{sec:main_proof}

We consider now the same basic technique as in the proof of Proposition \ref{prop:Velo}, replacing, however, $\lambda$ by $-i\lambda$ for $\lambda\in\Rl^s$. Schematically replacing $\lambda$ by $-i\lambda$ in \eqref{Wlambda} and \eqref{Wlambdap} we get
\begin{align}\label{Wilambda}
  W_{-i\lambda} &= e^{\lambda\cdot Q}W e^{-\lambda\cdot Q},        \\
  \intertext{and}
  \bigl(W_{-i\lambda}\psi\bigr)(p) &= W(p+i\lambda)\psi(p),      \label{Wilambdap}
\end{align}
but we still have to make rigorous sense of these identites.
In \eqref{Wilambdap} the right hand side contains the analytically continued $W(p)$, which exists by Assumption \ref{assumption:analytic_ext}. Hence $W_{i\lambda}$ can and will be defined by this formula. The $p$-dependent matrix it multiplies with is continuous, and since $\Brill$ is compact,  $W(p+i\lambda)$ is uniformly norm bounded for $p\in\Brill$. Hence $W_{i\lambda}$ is a bounded operator, and the family $W_z$ for $z\in\Cx$ is a norm analytic family of bounded operators. However, to make good use of this analytic extension we also need to give an interpretation of the operator product in \eqref{Wilambda}. This is done in the following Lemma.

\begin{lem}\label{lem:itwine}
Let $\Lambda>0$ and denote by $\DD_\Lambda$ the intersection of the domains of the selfadjoint operators $e^{\lambda\cdot Q}$ with $\norm\lambda\leq\Lambda$.   Then, for $\psi\in\DD_\Lambda$ we have $W\psi\in\DD_\Lambda$, and for $\norm\lambda<\Lambda$
\begin{equation}\label{itwine}
  e^{\lambda\cdot Q}W\psi=W_{-i\lambda}e^{\lambda\cdot Q}\psi.
\end{equation}
\end{lem}

\begin{proof}
For $\psi\in\DD_\Lambda$, the function $z\mapsto e^{iz\cdot Q}\psi$ can be differentiated in norm for $\norm{\im z}<\Lambda$, and has continuous boundary values on this ``strip''.  Consider now a vector $\phi$, which is compactly supported in position, so that, in particular,
$\phi\in\DD_\Lambda$.
Then consider the equation for complex $z\in\Cx^s$:
\begin{equation}\label{Wila}
  \braket{e^{-i\overline{z}\cdot Q}\phi}{W\psi}=\braket{\phi}{W_z e^{iz\cdot Q}\psi}.
\end{equation}
This holds for $z\in\Rl^s$ by \eqref{Wlambda}. Hence it can be extended to all $z$ for which both sides are analytic. The vector $e^{-i\overline{z}\cdot Q}\phi$ is everywhere anti-analytic, so the left hand side is entire. For the right hand side, we just established that the vector $e^{iz\cdot Q}\psi$ is analytic in the strip $\norm{\im z}<\Lambda$, and $W_z$ is an entire analytic family by assumption. The equality furthermore extends to the boundary values, in particular to any $z=-i\lambda$ with $\lambda\in\Rl^s$ and $\norm\lambda\leq\Lambda$. Thus, for $\phi$ in a core of the selfadjoint operator  $A=e^{\lambda\cdot Q}$ we have that $\braket{A\phi}{W\psi}=\braket{\phi}{W_{-i\lambda}A\psi}$. Hence $W\psi\in\dom(A^*)=\dom(A)$, and $AW\psi= W_{-i\lambda}A\psi$.
\end{proof}

The core of the proof of Theorem \ref{mainprop} is the following exponential estimate for the expectation of $\exp(\lambda\cdt Q)$:
\begin{lem}
Let $\rho$ be a density operator and $\Lambda>0$ such that $\abs{\tr\rho\exp(\lambda\cdt Q)}<\infty$ for $\lambda\in\Rl^s$ with $\norm\lambda\leq\Lambda$. Then, for such $\lambda$,
\begin{equation}\label{R-estimate}
 \limsup_{t\to\infty}\frac1t\log \tr\rho\, e^{\lambda\cdt Q(t)} \leq R(\lambda) ,
\end{equation}
where $R$ is the function given in Theorem \ref{mainprop}.
\end{lem}

Note that here we need a condition on the initial distribution, whereas in Propopsition \ref{prop:Velo} this is missing: The initial distribution is absorbed by the ballistic scaling. The condition is, of course, satisfied for any strictly localized state, but one can also find states, say, with power law decay, for which the expectation under the limit is constant equal to $+\infty$.

\begin{proof}
Note first that the condition is only about the initial distribution of $\rho$. In this condition $\tr\rho\exp(\lambda\cdt Q)$ is taken as always defined, but possibly infinite, by integrating the $\rho$-expectation of the spectral measure of $Q$ against the exponential function. We argue first, that it suffices to prove the Lemma for pure states $\rho=\kettbra\psi$. Indeed, if the assumption holds for $\rho$ it also holds for  every convex component of $\rho$, i.e., every vector $\psi$ such that $\rho\geq c\kettbra\psi$ with $c>0$. Furthermore, if the bound holds for some $\rho_i$'s it also holds for finite sums $\rho^N=\sum_i^N\rho_i$, where we can ignore normalization factors, which are scaled away on the left hand side with $1/t$. Moreover, we can transfer to a norm convergent sum $\rho=\sum_i^\infty\rho_i$, because $\tr\rho\exp(\lambda\cdt Q(t))$ is the supremum of the corresponding expressions using the partial sums $\rho^N$, and $R$ on the right hand side is independent of $\rho$. Hence the spectral resolution of $\rho$ reduces the proof to the pure state case.

For $\rho=\kettbra\psi$, the premise is just $\psi\in\DD_{\Lambda/2}$ with $\DD_{\Lambda/2}$ defined as in Lemma~\ref{lem:itwine}. That Lemma was stated in such a way that its condition propagates with $t$, i.e., it implies (by induction) also that $W^t\psi\in\DD_{\Lambda/2}$, and $e^{\lambda\cdot Q}W^t\psi=W_{-i\lambda}^te^{\lambda\cdot Q}\psi$. But then we find
\begin{align}\label{estpure}
  \tr\rho\, e^{\lambda\cdt Q(t)}&=\norm{e^{\lambda\cdt Q/2}W^t\psi}^2\nonumber\\
      &=\norm{W_{-i\lambda/2}^t\,e^{\lambda\cdt Q/2}\psi}^2\nonumber\\
      &\leq\norm{W_{-i\lambda/2}^t}^2\,\norm{e^{\lambda\cdt Q/2}\psi}^2
\end{align}
Hence taking the logarithm and dividing by $t$, we get that
\begin{equation}\label{R-estimate2}
 \limsup_{t\to\infty}\frac1t\log \tr\rho\, e^{\lambda\cdt Q(t)} \leq \limsup_{t\to\infty}\frac1t\log\norm{W_{-i\lambda/2}^t}^2.
\end{equation}
In fact, on the right hand side the limit (not just the limit superior) exists, because $t\mapsto \log\norm{A^t}$ is a subadditive function.  
For discrete time, $\lim_t\norm{A^t}^{1/t}$ is defined as the spectral radius of $A$, so in \eqref{R-estimate2} we get twice the logarithm of the spectral radius of $W_{-i\lambda/2}$. In continuous time, the limit along $t\in\Nl$ can similarly be evaluated, and by the spectral mapping theorem gives $2\sup\im\:\spec H_{-i\lambda/2}$. The same holds along $t_0\Nl$ for any $t_0$, so this is the limit. In either case, the operators involved are matrix multiplication operators, so the spectrum is the closure of the union over $p\in\Brill$ of the spectra of $W_{-i\lambda/2}(p)=W(p+i\lambda/2)$, and similarly for $H$.  This spectral supremum is the function $R$ given in Theorem~\ref{mainprop}.
\end{proof}

Our next task is to turn this into an estimate of probabilities of sets $M$. We begin with half spaces
\begin{equation}\label{halfspace}
    M=E_\geq(c,\lambda)=\{x\mid \lambda\cdot x\geq c\},
\end{equation}
for some constants $\lambda\in\Rl^s$ and $c\in\Rl$.
Then, for all $t\geq 0$, the indicator function of $M$ is everywhere $\leq\exp(t\lambda\cdot x-ct)$.
Therefore, by the functional calculus for the self-adjoint operator $Q(t)$,
\begin{equation}\label{ilexp}
    \chi_M(Q(t)/t)\leq e^{\lambda\cdot Q(t)-ct}
\end{equation}
and with \eqref{pmdef}:
\begin{equation*}
     p_t(\rho,M)
        \leq e^{-ct}\ \tr\rho\, e^{\lambda\cdot Q(t)}.
\end{equation*}
Combining with \eqref{R-estimate}, we can summarize this as 
\begin{equation}\label{est1}
     \limsup_{t\to\infty} \frac1t\, \log p_t(\rho,E_\geq(c,\lambda))
        \leq-c +\rateL(\lambda).
\end{equation}

Our task in the remainder of the proof will be to extend this from half spaces to more general sets, for which we largely follow the proof of \cite[Lemma VII.4.1.]{Ellis}. Throughout we will fix some $a>0$, to be thought of as a reference rate, and ask for which sets $M$ the probability decreases at least with rate $-a$, i.e., $\limsup_t t^{-1}\log p_t(\rho,M)\leq-a$.
The largest half spaces for which \eqref{est1} guarantees such decrease are 
\begin{equation}\label{Egeq}
  E_\geq(\lambda):=E_\geq(a+R(\lambda),\lambda),
\end{equation}
where we suppress the dependence on the fixed $a$ in the notation. It will be convenient to consider also the open half space $E_{>}(\lambda)=\{x\mid \lambda\cdot x>a+\rateL(\lambda)\}$, and their complements $E_{\leq}(\lambda)$. 

The first extension is to sets $M$ which are contained in a finite union of open half spaces,
\begin{equation}\label{finUnion}
  M\subset\bigcup_{i=1}^n E_{>}(\lambda_i).
\end{equation}
Indeed, in this case
\begin{equation*}
  p_t(\rho,M)\leq n \max_ip_t\bigl(\rho,E_\geq(a+\rateL(\lambda_i),\lambda_i)\bigr),
\end{equation*}
and hence
\begin{equation}\label{Mbound}
 \limsup_{t\to\infty} \frac1t\, \log p_t(\rho,M)\leq -a.
\end{equation}

Let us denote by $E_{>}=\bigcup_\lambda E_{>}(\lambda)$ the union of {\it all} open half spaces (for fixed $a$). Its complement is the \emph{level set} $\level(a)$ defined as
\begin{align}\label{E>c}
    (E_{>})^c&=\bigcap_\lambda\{x\mid \lambda\cdot x-\rateL(\lambda)\leq a\} \nonumber\\
         &=\{x\mid \sup_\lambda(\lambda\cdot x-\rateL(\lambda))\leq a\} \nonumber\\
        &=\{x\mid \ratef(x)\leq a\}=:\level(a)
\end{align}
of the rate function $\ratef$ defined by \eqref{Legendre}. We claim that for every closed set $M\subset E_{>}$ we can find a finite collection of $\lambda_i$ satisfying \eqref{finUnion}. This is evident for compact $M$, because we may choose a finite subcover of the cover $M\subset\bigcup_\lambda E_{>}(\lambda)$.

For general closed $M$ we follow a construction of \cite[Lemma~VII.4.1]{Ellis} exploiting that in a union of half spaces infinity is anyhow well covered. More formally, let $B$ be an open ball containing $\level(a)$, and denote its closure and boundary by $\overline B$ and $\partial B$ (see Figure \ref{fig:polygon}). Consider the set
\begin{equation}\label{Mtilde}
    \widetilde M=(M\cap\overline B)\cup \partial B.
\end{equation}
\begin{figure}[ht]
    \centering
  \includegraphics[width=0.3\textwidth]{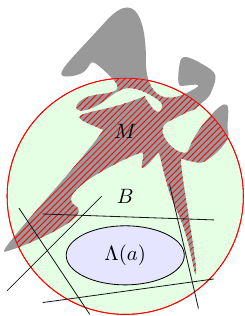}
  \caption{Subsets needed in the proof of the basic estimate. The red hatched patch is $\widetilde M$ from \eqref{Mtilde}. The straight lines indicate the half spaces characterized by the $\lambda_i$.}
  \label{fig:polygon}
\end{figure}
This is compact as the union of two compact sets, so we can find finitely many $\lambda_i$ with $\widetilde M\subset\bigcup_{i=1}^n E_{>}(\lambda_i)$. The complement $P=\bigcap_iE_{\leq}(\lambda_i)$ is a finite intersection of half spaces, i.e., a polytope. Since $\partial B$ is in the complement of $P$, $P$ is entirely contained in $B$. Hence $B^c\subset \bigcup_{i=1}^n E_{>}(\lambda_i)$, and since
$M\subset\widetilde M\cup B^c$, we have also found the desired covering \eqref{finUnion}. To conclude: \eqref{Mbound} holds for every closed subset $M\subset E_{>}$.

To conclude the proof of Theorem \ref{mainprop}, let $M$ be an arbitrary closed set and $I(M):=\inf_{x\in M}\ratef(x)$. As we argue in the next section, $\ratef$ is lower semi-continuous with compact level sets and hence this infimum is actually attained on $M$. Since for $\ratef(M)=0$ the bound in \eqref{LDbound} is trivial, we assume $0<\ratef(M)\leq\infty$. Consider any $a<\ratef(M)$, which for $I(M)=\infty$ just means any positive number. The level set $\level(a)$ is disjoint from $M$, and so $M\subset E_{>}$. By the previous paragraph this implies \eqref{Mbound}. Since $a<\ratef(M)$ was arbitrary we finally get the bound \eqref{LDbound}.

\section{Discussion of the bound and general properties}\label{sec:discussion}
\subsection{Elementary properties}
\ititem1{$0\leq\ratef(x)\leq\infty$}
Obviously,  by \eqref{formula} and \eqref{formulaW}, $\rateL(0)=0$. Hence we can put $\lambda=0$ in \eqref{Legendre} to get $\ratef(x)\geq0$. The rate function can be infinite, which indicates superexponential decay. A good example is given also in Section \ref{sec:largeX}, where in some velocity region the probability vanishes exactly, and $I=\infty$.

\ititem2{$\ratef$ is convex and lower semicontinuous}
Convexity and lower semicontinuity are clear for any Legendre transform, and more generally for any pointwise supremum of continuous affine functions.

\ititem3{$I(x)=0$ for $x\in\co\Gamma$, the closed convex hull of the propagation region.}
Indeed, if we have a set such that $p_t(\rho,M)\geq q>0$ for some $\rho$ and large $t$, the left hand side of \eqref{LDbound} has a lower bound $(\log q)/t)$, which goes to zero, implying $\inf_{x\in M}\ratef(x)=0$. Since $\ratef$ is lower semicontinuous, this implies that $I(x)$ vanishes on $\Gamma$. Since the rate function $I$ is convex, its level sets and in particular the zero set is also convex, so contains $\co\Gamma$. The propagation region may well fail to be convex, in which case there is a region for which our bound has no exponential rate prediction, although the probability does go to zero. An example is provided in Section \ref{sec:non_conv} below.

\ititem4{$\rateL$ is continuous}
Consider the function $(p,\lambda)\mapsto \spr W(p+i\lambda)$, where `$\spr$' denotes the spectral radius. By our standing assumption, $W(p)$ is entire analytic, so jointly continuous. Moreover, the spectral radius in a fixed finite matrix dimension is continuous, because the coefficients of the characteristic polynomial are, and the zeros of a polynomial are continuous functions of the coefficients (a standard result; see e.g.,  \cite{HarrisMartin}). The continuity of $R$ thus follows from the purely topological Lemma that for any jointly continuous function $f:X\times Y\to \Rl$ with $Y$ compact, the partial supremum $g(x)=\max_{y\in Y}f(x,y)$ is continuous \cite{WWong}. For completeness, we provide the lemma in Appendix \ref{app:sup_cont}.

\ititem5{$\ratef$ has compact lower level sets}
The lower level sets
\begin{equation}\label{levelset}
    \level(a)=\{x\mid \ratef(x)\leq a\}
\end{equation}
defined in \eqref{E>c} are closed for all $a\geq0$ because $\ratef$ is lower semicontinuous. To show their boundedness we only need one consequence of the continuity of $\rateL$, namely that $\rateL$ is bounded on any sphere of sufficiently small radius $\varepsilon$.  Suppose $\abs{R(\lambda)}\leq c$ for $\abs\lambda=\varepsilon$. Then, for $x\in\Lambda(a)$, we have $\lambda\cdot x-R(\lambda)\leq a$ for all $\lambda$, and hence
$\lambda\cdot x\leq(a+c)$ for $\abs\lambda=\varepsilon$. Hence $\level(a)$ is contained in a ball of radius $(a+c)/\varepsilon$.

\subsection{Large $x$ and relation to trivial propagation bounds}\label{sec:largeX}

Consider first the walk case, i.e., the discrete time evolution with strictly finite jumps, characterized by $W(p)=\sum_{y\in F}W_y\expi{p\cdot y}$, where $F=\{y\mid W_y\neq0\}$ is the finite {\it set of possible jumps}.
\begin{lem}\label{lem:Iinf}
We have $I(x)=\infty$ for all $x$ outside the convex hull of the possible jumps $\co F$.
\end{lem}

\begin{proof}  Introduce the function $\rateL_F(\lambda)= \sup\{\lambda\cdot x|x\in F\}$. Its Legendre transform is
\begin{equation}\label{rategauge}
   \ratef_F(x) = \begin{cases}0& x\in\co F\\\infty &\text{otherwise}.\end{cases}
\end{equation}
With $c=\log\sum_y\norm{W_y}$ we have
\begin{equation}\nonumber
    \spr W(p+i\lambda)\leq\norm{W(p+i\lambda)}\leq \Bigl(\sum_y\norm{W_y}\Bigr)\ \exp(\rateL_F(\lambda))
                        \equiv\exp\bigl(c+\rateL_F(\lambda)\bigr),
\end{equation}
and, with \eqref{formulaW} and \eqref{Legendre},
\begin{align*}
   \rateL(\lambda)\leq 2(c+\rateL_F(\lambda/2)),\qquad
    \ratef(x)\geq 2\ratef_F(x)-c=\infty,
\end{align*}
whenever $x\notin\co F$.
\end{proof}

The continuous time analogue of the last bound is somewhat similar to Lieb-Robinson bounds in statistical mechanics \cite{liebFiniteGroupVelocity1972,nachtergaeleLiebRobinsonBoundsQuantum2010,naaijkensQuantumSpinSystems2017}: One needs only rough information about the size of the Hamiltonian terms to get propagation in a cone (up to exponential tails). We describe one simple version. To state it, we assume without loss that the origin lies in the interior of $\co F$ with $F=\{y\mid H_y\neq0\}$, where $H(p)=\sum_{y\in F}H_ye^{ip\cdot y}$. (If this is not the case, add ${\idty}$ at the origin.)
The so-called ``gauge functional'' of this set is
\begin{equation}\label{gaugefct}
    g(x)=\inf\{\alpha\in \mathbb R^+_*\mid x\in\alpha\, \co F\}
\end{equation}
This is like a norm, but with a ``unit ball'' which need not be symmetric around $0$. Then
\begin{lem}
With $C=\sum_y\norm{H_y}$,
\begin{equation}\label{LiebR}
    \ratef(x)\geq 2 g(x)\,\log\frac{g(x)}{eC}.
\end{equation}
\end{lem}
Note that this bound is non-trivial only for $g(x)>e C$, but then grows a bit faster than linearly. Moreover, it sets a bound on the propagation region, namely
\begin{equation}\label{propHam}
    \Gamma\subset eC\ \co F.
\end{equation}

\begin{proof} Let $\omega$ be a complex eigenvalue of $H(p+i\lambda/2)$, and $\psi$ a corresponding normalized eigenvector.
Then
\begin{align*}
  \im\:\omega&=\im\:\bra\psi H(p+i\lambda/2)\ket\psi  \\
           &\leq \sum_y e^{y\cdot\lambda/2}\im\left(\braket\psi{H_y\psi}e^{-iy\cdot p}\right)\\
           &\leq \max_{y\in F} e^{y\cdot\lambda/2}\ \sum_y\norm{H_y} \\
           &= C \exp\left( \rateL_F(\lambda/2)\right),
\end{align*}
which we can summarize as
\begin{equation}\label{LRrateL}
    R(\lambda)\leq 2C\, \exp(\rateL_F(\lambda/2)),
\end{equation}
where $\rateL_F$ is defined as in the proof of Lem. \ref{lem:Iinf}.
Now suppose that $g(x)>\alpha$, i.e., $x\notin \alpha\co F$. Then there is some $\lambda_*$ such that $\rateL_F(\lambda_*)\leq1$, but $\lambda_*\cdot x\geq\alpha$. We will get a lower bound on $\ratef(x)$, by extending the supremum \eqref{Legendre} only over subset
of $\lambda=2t\lambda_*$, where $t\geq0$. Thus
\begin{align*}
  \ratef(x)&\geq\sup_{t\geq0}\Bigl\{2t\lambda_*\cdot x -2C\,\exp(\rateL_F(t\lambda_*))\Bigr\}  \nonumber\\
       &\geq 2 \sup_{t\geq0}\bigl\{\alpha t -C\,e^t\bigr\}\nonumber\\
       &= 2\alpha(\log\frac\alpha C-1), \nonumber
\end{align*}
where we have assumed that $\alpha>C$, since otherwise the supremum is attained at $t=0$, giving the trivial bound $\ratef(x)\geq-C$.
Since this bound is monotone in $\alpha<g(x)$, we may replace $\alpha$ by $g(x)$, which yields \eqref{LiebR}.
\end{proof}

\subsection{Small $\lambda$ and behavior near $\partial\Gamma$}
In order to analyze the behavior of the rate function $I$ near the boundary of the propagation region $\Gamma$, we consider the behavior of $\rateL$ for small $\lambda$. We first note that as a consequence of the identity $W^{-1}(p+i\lambda)=W^*(p-i\lambda)$, any isolated eigenvalue $e^{-i\omega_j(p+i\lambda)}$ satisfies
\begin{equation}
\im\:\omega_j(p+i\lambda)=-\im\:\omega_j(p-i\lambda).
\end{equation}
Therefore, around any regular point $p$ the even powers in the expansion of this expression vanish, and we can write: 
\begin{equation}\label{expnw1}
    \im\:\omega_j(p+i\lambda)=\lambda\cdot\nabla\omega_j(p)
         +\Order(\lambda^3)
\end{equation}
To obtain an expression for $R$, according to \eqref{formula} we have to take the maximum of this with respect to $p$. Doing this in the approximation \eqref{expnw1} by \eqref{specV} directly leads to $\rateL(\lambda)\approx\rateL_\Gamma(\lambda)$, with $\rateL_\Gamma$ defined as above by $\rateL_\Gamma(\lambda)= \sup\{\lambda\cdot x|x\in \Gamma\}$. The corresponding rate function $\ratef_\Gamma$ vanishes on $\Gamma$, as $\ratef$ does, but is infinite everywhere outside $\co\Gamma$. The approximation \eqref{expnw1} is thus useless for learning anything about the behavior of $\ratef$ near the boundary $\partial\Gamma$. 

We will thus go to higher orders in the Taylor expansion \eqref{expnw1}. However, it is clear from the outset that such an expansion cannot lead to a full evaluation of $\rateL$ in general: This is defined as the global supremum of a function that is periodic in $p$, and the expansion will destroy that property and introduce artifacts, especially in the polynomial growth for large $p$. Similarly, the Legendre transform to the rate function $\ratef$ can only be estimated from below from any local description of $\rateL$. Therefore, the best we can hope for is a heuristic description of the typical behaviour near the boundary.

So let us fix $\lambda$ and start from the maximization of \eqref{expnw1} with respect to $p$ and $j$. This is the same as the maximization of $n\cdot\nabla\omega_j(p)$ for the dual direction $n=\lambda/\abs\lambda$, i.e., exactly the same variational problem we need for determining the boundary points of $\co\Gamma$. We assume a unique solution, as will be generically the case, and accordingly fix $p$ and $j$ in the sequel.  $p$ will be the base point for the Taylor expansion of $\omega_j$. 
We note that these data are generally different for $\lambda$ and for $-\lambda$. The expansion then reads
\begin{equation}\label{expnw3}
\begin{aligned}
  \im\,\omega_j(p+i\lambda+\delta p)=\omega_{j|\alpha}\lambda_\alpha\,+\mskip-200mu&\\
  &+\frac16\ \omega_{j|\alpha\beta\gamma}\lambda_\alpha\Bigl(-\lambda_\beta\lambda_\gamma+3\delta p_\beta\,\delta p_\gamma\Bigr)\\
  &+ \frac1{6}\ \omega_{j|\alpha\beta\gamma\eta}\lambda_\alpha\Bigl(-\lambda_\beta\lambda_\gamma
            +\delta p_\beta\,\delta p_\gamma\Bigr) \delta p_\eta
    +\cdots
\end{aligned}
\end{equation}
Here we use the shorthand $f_{\vert\alpha}:=\partial f/\partial p_\alpha$ and the summation convention by which repeated greek indices are understood to be summed over. The combinatorial factors arise from a $1/n!$ in the Taylor expansion and the observation that the derivatives are permutation symmetric in the indices, so terms with the same number of $i\lambda$- and $\delta p$-components but in different positions are equal. Moreover, the term with second derivatives of $\omega_j$ is left out preemptively: It vanishes because $\partial_\beta(\lambda_\alpha\omega_{j|\alpha})$ vanishes at the extremum. For the same reason, in the third order term the matrix $m_{\beta\gamma}=\lambda_\alpha\omega_{j|\alpha\beta\gamma}$ is negative semidefinite.

In order to evaluate $\rateL(\lambda)$ via \eqref{formula} we need to compute the maximum of this expression over $\delta p$. For the dominant contribution, the term in the second line of \eqref{expnw3}, this is easy: By choice of $p$ its maximum is at $\delta p=0$. We expect a small correction to this when higher order terms are included. Setting the derivative of the right hand side of \eqref{expnw3} to zero leaves leaves the  linear term $m_{\beta\gamma}\delta p_\gamma$ from the second line and, from the third, a constant term $\propto\lambda^3$ plus a term $\propto(\lambda\delta p)^2$, which we neglect. Assuming the matrix $m$ to be invertible, we can solve for $\delta p$ and get a correction $\delta p\propto \lambda^2$. Inserting that back into the right hand side we get a correction $\propto \lambda^4$, which is neglected in comparison to the higher order terms that we left out in the first place. Thus we conclude that to order $\lambda^3$ we can set $\delta p=0$.

Let us introduce a variant of the function $\rateL$, a ``boundary version'', in which only this conclusion is realized, i.e. $R_b(\lambda)=\im\,\omega_j(p+i\lambda)$, where $(j,p)$ are chosen to maximize $\lambda\cdot\nabla\omega_j(p)$. Then we can summarize the above heuristic discussion as
\begin{equation}\label{Rboundary}
  R(\lambda)\geq R_b(\lambda)=R(\lambda)+\Order(\lambda^4)=
            \omega_{j|\alpha}\lambda_\alpha
             -\frac16\ \omega_{j|\alpha\beta\gamma}\lambda_\alpha\lambda_\beta\lambda_\gamma +\Order(\lambda^4)
\end{equation}
We will use this for computing the Legendre transform for $\ratef$ half-ray by half-ray, i.e., along sets $\lambda=tn$, $t\geq0$ for some unit vector $n$. Then $R(tn)=a t+bt^3+\Order(t^4)$ with $a\geq0$, and $b\geq0$ due to the negative definiteness of $m$. Thus we need

\begin{lem}\label{lem:bdary}
Let $a\geq0$, $b>0$, and  $R:\Rl_+\to\Rl$ be such that
$$R(t)=at+bt^3 + \Order(t^4)\quad.$$
Consider its Legendre transform $I(x)=\sup_{t\geq0}\{xt-R(t)\}$. Then for $x\geq a$:
\begin{equation}\label{LegThreeHalf}
    I(x)\geq\frac{2}{\sqrt{27\,b}}\ (x-a)^{3/2}+\Order((x-a)^{2})\quad.
\end{equation}
\end{lem}

\begin{proof}
For $x>a$ we have to maximize a function of the form $(x-a)t-bt^3$, which has a unique maximum for $t\geq0$ at $t_m=\sqrt{(x-a)/3b}$, where its value is the first term on the right hand side of \eqref{LegThreeHalf}.  To estimate the error term, set $R(t)=at+bt^3 + r(t)$ with $\abs{r(t)}\leq c t^4$ for $t\leq t_c$. Then for $(x-a)\leq 3bt_c^2$ we have $t_m\leq t_c$, so the bound on $r$ applies at $t_m$, and 
\begin{equation*}
  I(x)\geq (x-a)t_m-bt_m^3-r(t_m)\geq \frac2{\sqrt{27b}}(x-a)^{3/2} - c(3b)^{-2}\ (x-a)^2.
\end{equation*}
\end{proof}

Note that when we  just assume $R$ to be known asymptotically for $t\to0$ we cannot do better, since the true maximum could be attained anywhere on the half axis. However, if we add the assumption that $R$ is strictly convex, the local maximum in the proof must be the unique maximum, and  we get an asymptotic equality as $(x-a)\to0$.

Let us summarize the heuristic estimate for the boundary behaviour of the rate function $\ratef$. Every boundary point will contribute a lower bound. More precisely, each such bound is associated with a normal direction $n$ and a boundary point $x_*=\nabla\omega_j(p)$, with $(j,p)$ chosen to maximize $n\cdot\nabla\omega_j(p)$. We now assume that \eqref{Rboundary} holds and apply Lemma~\ref{lem:bdary}. The argument $x$ in the Lemma will be $n\cdot x$, $a=n\cdot x_*$, and 
\begin{equation}\label{bboundary}
  b=\omega_{j|\alpha\beta\gamma}n_\alpha n_\beta n_\gamma.
\end{equation}
Then 
\begin{equation}\label{Iboundary}
   I(x)\geq \frac{2}{\sqrt{27\,b}}\ \bigl(n\cdot(x-x_*)\bigr)_+^{3/2} + \Order\bigl((n\cdot(x-x_*))_+^2\bigr)
\end{equation}
Here $z_+$ denotes the positive part of $z\in\Rl$, and we use this to express that the bound is non-trivial only for the points 
with $(n\cdot(x-x_*))>0$, for which surely $x\notin\co\Gamma$. For the overall bound one can take the supremum over all $(x_*,n,p,j)$.

\section{Examples}\label{sec:Exmp}

\subsection{1D Qubit walk}
The asymptotics of quantum walks has been mostly investigated in the allowed region: in the one-dimensional setting \cite{konno2005new,carteret}, in the two-dimensional setting \cite{PhysRevA.77.062331,baryshnikovTwodimensionalQuantumRandom2010} and in arbitrary dimension \cite{Scudo,timerandom}. To our best knowledge, only two studies \cite{Carteret2,SunaTa} have been made of the asymptotics outside of the propagation region. In these studies a detailed stationary phase analysis was carried out involving the choice of contours in the complex plane. In addition to the exponential decay \eqref{LDbound}, this also gives the prefactors. Hence our method provides less information, but is much simpler to apply, and therefore has a better chance to analyze more complex cases. The simplest walks are of the form
\begin{equation}\label{1Dsimple}
    W(p)=\left(\begin{array}{cc}e^{ip}&0\\0&e^{-ip}\end{array}\right)
         \left(\begin{array}{cc}a&b\\-\overline b&\overline a\end{array}\right),
\end{equation}
with $\abs a^2+\abs b^2=1$.
Since $\det W(p)=1$, the dispersion relations are entirely determined by $\tau(p)=\frac12\tr W(p)=\bigl(ae^{ip} +\overline ae^{-ip}\bigr)/2$. The phase of $a$ can be compensated by a shift in $p$, which is irrelevant for our question, so we assume from now on that $a=\abs a>0$. Then $W(p)$ has the eigenvalues $\omega_\pm(p)=\pm\arccos(a \cos(p))$. This gives the group velocities
\begin{equation}\label{1Dgroupv}
    v_\pm(p)=\frac{\pm a \sin (p)}{\sqrt{1-a^2 \cos ^2(p)}},
\end{equation}
which takes their extrema $\pm a$ for $p=\pm\pi/2$. So these are the boundary points of the propagation region.
The maximum of $\abs{\exp(-i(\omega_\pm(p+i\lambda)))}$ for $\lambda\in\Rl$ is also assumed at $p=\pm\pi/2$. Indeed, the derivative of this expression with respect to $p$ equals zero iff $v_\pm (p+i\lambda) =\overline{v_\pm (p+i\lambda)}$, which, after some algebra, yields the $\lambda$-independent equivalent condition $\sin (p) \cos (p)=0$. Hence in this case we have $R(\lambda)=R_b(\lambda)$ with $R_b$ as in \eqref{Rboundary}. With this information we directly get from \eqref{formulaW}:
\begin{equation}\label{1Dprerate}
    \rateL(\lambda)=2\arsinh\Bigl(a \sinh\bigl(\abs{\lambda/2}\bigr)\Bigr).
\end{equation}
It turns out that the equation $\rateL'(\lambda)=x$ can be solved explicitly for $\lambda$, which gives
\begin{equation}\label{1Dlambda}
    \lambda(x)=2\log \left(\frac{\sqrt{x^2-a^2}+x\,\sqrt{1-a^2}}{a \sqrt{1-x^2}}\right),
\end{equation}
for $1\leq a\leq x\leq1$, and in this range the Legendre transform
\begin{equation}\label{1Drate}
    \ratef(x)= x \lambda(x)
                 -2\log \left(\frac{\sqrt{x^2-a^2}+\sqrt{1-a^2}}{\sqrt{1-x^2}}\right).
\end{equation}
These functions are displayed in Fig.~\ref{fig:1DRI}.
This coincides with  \cite[(1.17)]{SunaTa}, but disagrees with \cite{Carteret2}.
\begin{figure}[t]
  \centering
  \includegraphics[width=\textwidth]{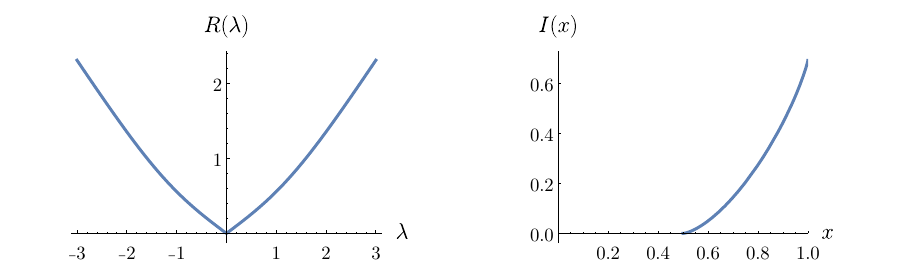}
  \caption{The functions $\rateL$ and $\ratef$ after \eqref{1Dprerate} and \eqref{1Drate} for $a=.5$.}
  \label{fig:1DRI}
\end{figure}
We note for later purposes the expansions near the boundary which follow \eqref{Rboundary} and \eqref{Iboundary}, i.e.,
\begin{align*}
  \rateL(\lambda)&=a\lambda+\frac16\,(a-a^3)\lambda^3+\Order(\lambda^4)\\
  \ratef(x)      &=\sqrt{\frac83}\frac{(x-a)_+^{3/2}}{\sqrt{a-a^3}} +\Order((x-a)_+^{2}).
\end{align*}

\subsection{A 2D walk with circular light cone}
In a one-dimensional system the branches $\omega_j$ can always be chosen locally to be analytic. In fact, a basic result of perturbation theory \cite[§1. Theorem 1.10]{Kato} states that even at degenerate points one can make this choice. In higher dimensions one aspect of this structure survives: along any straight line through a degenerate point $p_*$, say $t\mapsto p_*+t\delta p$, one can pick analytic branches, and in particular some discrete set of slopes $d\omega_j(p_*+t\delta p)/dt$. However, these slopes in general do not belong to several intersecting analytic functions (in first order an intersection of planes) but may instead form a cone. The example we give here is perhaps the simplest in which this happens. Moreover, the conical singularity is essential for determining the outer boundary of the propagation region $\Gamma$, which in this case is a disc. We set
\begin{equation}\label{W2DWeyl}
    W(p_1,p_2)=\exp(ip_1\sigma_1)\exp(ip_2\sigma_3)
    =\left(
\begin{array}{cc}
 e^{i p_2} \cos p_1 & i e^{-ip_2}\sin p_1 \\
 i e^{ip_2} \sin p_1 & e^{-i p_2} \cos p_1
\end{array}
\right),
\end{equation}
which is equivalent to the product of two one-dimensional walks of the form \eqref{1Dsimple} for $a=1/\sqrt{2}$.
Its dispersion relation is
\begin{equation}\label{2DWeylOmega}
    \omega_\pm(p_1,p_2)=\pm\arccos\bigl(\cos p_1\cos p_2\bigr)=:\pm\omega(p).
\end{equation}
\begin{figure}[t]
	\begin{center}
		\includegraphics[width=0.9\textwidth]{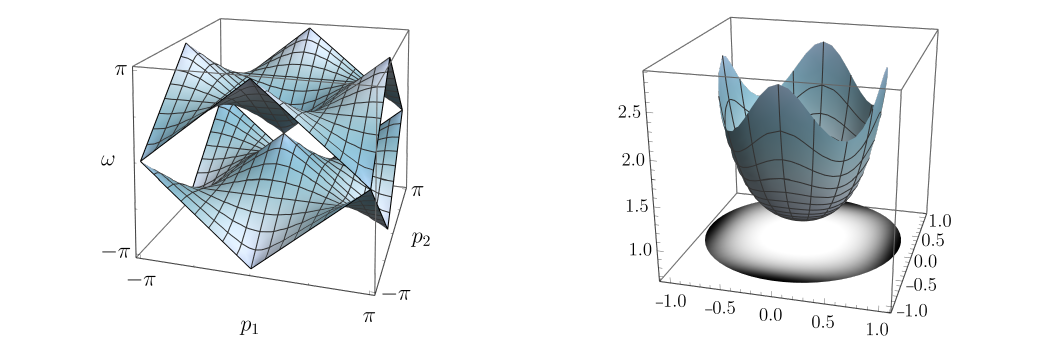}
		\caption{\label{fig:2DWeyl}Left: Dispersion relation $\omega_\pm(p)$ according to \eqref{2DWeylOmega}.
			Right: Probability density of the group velocity for particles starting at the origin. The mesh on the graph is in polar coordinates for velocity, shown up to $\abs v=.8$.
			Bottom right: Propagation region, i.e., the unit disc, shaded according to probability density.}
	\end{center}
\end{figure}
The conical structure at $p=0$ becomes apparent when we expand $\omega(\varepsilon p)$ in a Taylor series in $\varepsilon$:
\begin{equation}\label{omCone2}
    \omega(\varepsilon p_1,\varepsilon p_1) =\varepsilon\abs{p\,} -\frac{\varepsilon^3}6\,\frac{p_1^2\,p_2^2}{\abs{p\,}}
    -\frac{\varepsilon^5\ p_1^2\,p_2^2}{360\abs{p\,}^3}\Bigl(4p_1^4+13p_1^2\,p_2^2+4p_2^4\Bigr)
    +\Order(\varepsilon^7)\quad,
\end{equation}
where $\abs{p}=\sqrt{p_1^2+p_2^2}$. Of course, the same phenomenon happens at all points where $\cos(p_1)\cos(p_2)=\pm1$, i.e., $p_1,p_2=0,\pi$.
This qualitatively explains Figure \ref{fig:2DWeyl}. The group velocity can be determined directly from \eqref{W2DWeyl}:
\begin{equation}\label{2DWeylV}
    \nabla\omega_\pm(p)=\frac{\pm1}{\sqrt{1-\cos^2p_1\cos^2p_2}}
        \left(\begin{array}{r}\sin p_1 \cos p_2\\ \cos p_1\sin p_2\end{array}\right)
\end{equation}
with modulus
\begin{equation}\label{2DWeylV}
    \Bigl|{\nabla\omega_\pm(p)}\Bigr|^2=\frac{(1-\cos^2 p_1)\cos^2 p_2+\cos^2 p_1(1-\cos^2p_2)}{({1-\cos^2p_1 \cos^2p_2})}
\end{equation}
This is a monotone function of each $\cos^2p_i$, and equal to $1$, whenever one of these variables is equal to one. Hence the propagation region is contained in the disc $\abs{v\,}\leq1$. The boundary points are all attained, but apart from the points on the axes only in the limit $\abs{p\,}\to0$ (compare \eqref{omCone2}). The velocity distribution starting from the origin (i.e., with flat momentum distribution is shown in Figure \ref{fig:2DWeyl}.

\subsection{A 2D walk with non-convex propagation region}\label{sec:non_conv}
There is no reason why the propagation region $\Gamma$ should be a convex set. Figure \ref{fig:nonconvex} provides a somewhat minimal example. It is in the Hamiltonian context, without internal degree of freedom ($\dim\KK=1$)
and the Hamiltonian
\begin{equation}\label{nonconvexH}
  H(p_1,p_2)=\cos p_1 \: \cos p_2 -\sin p_1.
\end{equation}
\begin{figure}[ht]
    \centering
  \includegraphics[width=0.4\textwidth]{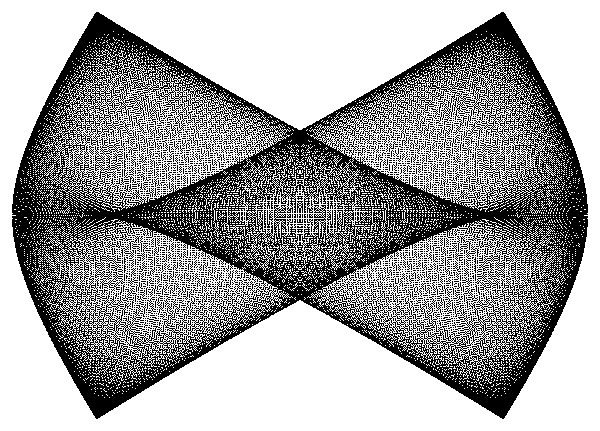}
  \caption{Propagation region for the Hamiltonian \eqref{nonconvexH}. The plot is generated by computing the group velocity on a regularly spaced grid in momentum space. The density of points thus corresponds to the asymptotic probability density for any reasonably well localized initial state.}
 \label{fig:nonconvex}
\end{figure}
For an example in the walk context, see \cite[Figure 3]{baryshnikovTwodimensionalQuantumRandom2010}.

Of course, it is to be expected that large deviation bounds depending on the distance from the boundary also hold inside the non-convex indentations. However, our methods do not provide such statements. They give a positive rate function only outside the convex hull of $\Gamma$ (recall that we approximate by half spaces).

\subsection{3D cones and rate functions}
Some quantum walks on the three-dimensional cubic lattice were proposed by Bialinycki-Birula  \cite{bb94} as discrete approximations to the Weyl, Dirac, and Maxwell equations. The interest in that paper is mainly in the continuum limit, and the possible variations of nearest neighbour walks automata for which this works. But one can take the discrete model as a system in its own right and analyze its propagation properties. Let us consider the simplest of these, the ``Weyl equation'' model in \cite{bb94}. We then have $\KK=\Cx^2$, and
\begin{equation}\label{bb94W}
    W(p_1,p_2,p_3)=e^{ip_1\sigma_1}\ e^{ip_2\sigma_2}\ e^{ip_3\sigma_3}.
\end{equation}
This gives
\begin{equation}\label{bb94omega}
    \omega_\pm(p)=\pm\omega(p)=\pm\arccos\Bigl(\cos p_1\cos p_2\cos p_3
                        -\sin p_1\sin p_2\sin p_3\Bigr).
\end{equation}
A characteristic feature here are conical singularities at $p=0$ and eight inequivalent points in the Brillouin zone, where  $p_1,p_2,p_3\in\pi\Ir$, meaning that all $\cos p_i=\pm 1$ and each matrix factor in \eqref{bb94W} is $\pm\idty$.
Expanding the cosine of \eqref{bb94omega} to second order in a small deviation $p$ from such a conical point $p^\vee$, we get, for any combination of the signs $\cos(p_k^\vee)$,
\begin{equation}\label{bb94omega_expansion}
    \omega(p^\vee+p)^2=  (p_1^2+p_1^2+p_1^2)+ \Order(p^3)
\end{equation}
This is in keeping with the remarks in Section \ref{sec:propregion} on regular vs.\ singular points in a dispersion relation. Along a straight line through $p^\vee$ we can use one-parameter perturbation theory, in which we can always choose analytic branches, corresponding to the straight lines in the boundary of the cone, but as a function of three parameters $\omega(p+p^\vee)\approx\omega(p^\vee)\pm\abs p$ is {\it not} analytic.
In the limit taken in \cite{bb94} only small momenta survive, so the propagation is indeed exactly the familiar light cone from relativistic physics. However, in this model faster propagation speeds can arise from points at a finite distance from the singularities. The overall propagation region is shown in Figure \ref{fig:3DWeylGamma}.
\begin{figure}[ht]
    \centering
  \includegraphics[width=0.4\textwidth]{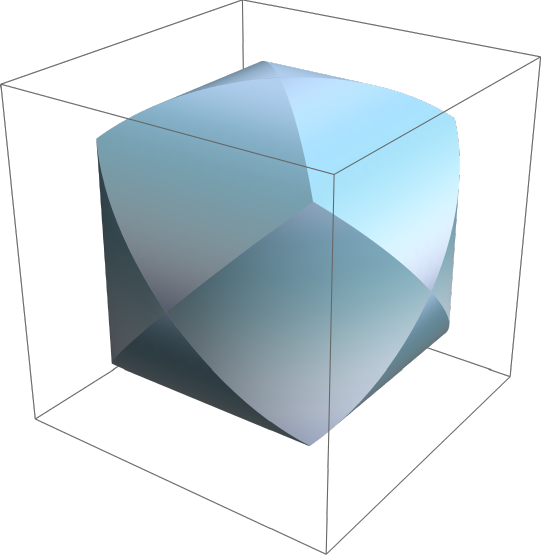}
  \caption{Propagation region for the walk \eqref{bb94W}.
    This convex set is the intersection of three orthogonal cylinders.}
  \label{fig:3DWeylGamma}
\end{figure}

To get a spherical propagation region, we consider a modified model, on the same lattice and also with $\KK=\Cx^2$, and with exactly the same conical points $p^\vee$, namely
\begin{align}\label{bb94spherical}
    W(p_1,p_2,p_3)&=\begin{pmatrix}
                      \cos p_1\,\cos p_2\ e^{ip_3} & -\cos p_1\,\sin p_2+ i\sin p_1 \\
                      \cos p_1\,\sin p_2+ i\sin p_1 & \cos p_1\,\cos p_2\ e^{-ip_3}
                    \end{pmatrix}\\
    \omega(p)&=\arccos\bigl(\cos p_1\cos p_2\cos p_3\bigr).\label{bb94threecos}
\end{align}
The modulus of the group velocity $\nabla \omega$ is then given by
\begin{align}\label{bb94sphericalom}
    \abs{\nabla\omega}^2&=\frac{\sin^2 p_1\cos^2p_2\cos^2p_3+\cos^2 p_1\sin^2p_2\cos^2p_3+\cos^2 p_1\cos^2p_2\sin^2p_3}{1-\cos^2 p_1\cos^2 p_2\cos^2 p_3} \nonumber\\
           &=1-\frac{(1-c_3)(1-c_1 c_2) + c_3(1-c_1)(1-c_2)}{1-c_1c_2c_3}
\end{align}
where, in the second line we have abbreviated $c_i=\cos^2p_i$. Since the fraction in the second line is manifestly positive, we get $\abs{\nabla\omega}\leq1$. The maximum is reached at the conical points $p^\vee$, where $c_1=c_2=c_3=1$.

In order to compute the rate function we first notice that the analytic continuation of $W$ to complex $p$ is still a $2\times2$-matrix with determinant $1$, so the trace (equal to $2\cos\omega(p)$) determines the secular equation and hence both eigenvalues. That is, all information we need to evaluate $\rateL(\lambda)$ from \eqref{formulaW} is contained in the analytic continuation of \eqref{bb94threecos}.

We further reduce the problem by not looking at the full rate function, but a radial version of it: $I(x)$ retains some of the lattice symmetry and is certainly not a rotation invariant function. As described in the introduction, we are however, mainly interested in an upper bound on the probability leakage outside a ball with slightly superluminal speed. We therefore look for a ``radial rate function'' $\ratefr$ such that $\ratef(x)\geq \ratefr(r)$ if $\abs x\geq r$. Then the bound \eqref{LDbound} for the probability outside a ball $M_{<r}$ of radius $r$ becomes
\begin{equation}\label{radialrate}
  \limsup_{t\to\infty} \frac1t\, \log p_t(\rho,M_{\geq r})\leq -\ratefr(r)
\end{equation}
which vanishes up to $r=1$. The Legendre transform is friendly to the radial reduction: Suppose that $\rateL(\lambda)\leq \rateLr(\ell)$, whenever $\abs\lambda\leq\ell$. Then, for $\abs x\geq r$,
\begin{align}\label{Iradial}
  \ratef(x) &=\sup_{\ell>0}\sup_{\abs\lambda\leq\ell} \bigl\lbrace x\cdot\lambda-\rateL(\lambda)\bigr\rbrace
     \geq \sup_{\ell>0} \sup_{\abs\lambda\leq\ell} \bigl\lbrace x\cdot\lambda-\rateLr(\ell)\bigr\rbrace\nonumber \\
    &= \sup_{\ell>0} \bigl\lbrace \abs x\ell-\rateLr(\ell)\bigr\rbrace
     \geq \sup_{\ell>0} \bigl\lbrace r\ell-\rateLr(\ell)\bigr\rbrace
     =:\ratefr(r),
\end{align}
which is the Legendre transform of $\rateLr$. Here the inequality signs were chosen to make $\ratefr$ and  $\rateLr$ automatically monotone.

This leaves us with the task to compute, for every $\ell\geq0$ the dual radial rate function
\begin{equation}\label{formulaWr}
  \rateLr(\ell)= \sup\Bigl\{\log\abs u^2 \Bigm| u \in\spec\Bigl(W\bigr(p+\tfrac{i\lambda}2\bigr)\Bigr)\
                    \text{for some}\ p,\lambda\in\Rl^3,\ \abs\lambda\leq\ell \Bigr\}
\end{equation}
Here it is useful to note that the eigenvalues are $u=\exp(i\omega)$ with $\omega\in\Cx$ determined from \eqref{bb94threecos} with complex substitution $p_k\mapsto p_k+i\lambda_k/2$, so that $\log\abs u^2=2\abs{\im\omega}$. The absolute value here takes care of the maximum over the spectrum (and the sign ambiguity for $\omega$). This makes $\log\abs u^2$ straightforward to evaluate for given $p,\lambda$.

The eigenvalue branches of $W$ are analytic in the complex vector variable $(p+i\lambda/2)$ except for the branch points at $p=p^\vee$, $\lambda=0$. Therefore, by the maximum principle \cite[Thm.~V.2.3]{Grauert}, no local maximum can occur away from these points and $\abs\lambda<\ell$. At the conical points the function is continuous and zero. The numerical evaluation of the maximum as a function of $\ell=\abs\lambda$ is shown in Fig.~\ref{fig:rateLr}.
\begin{figure}[ht]
    \centering
  \includegraphics[height=0.2\textheight]{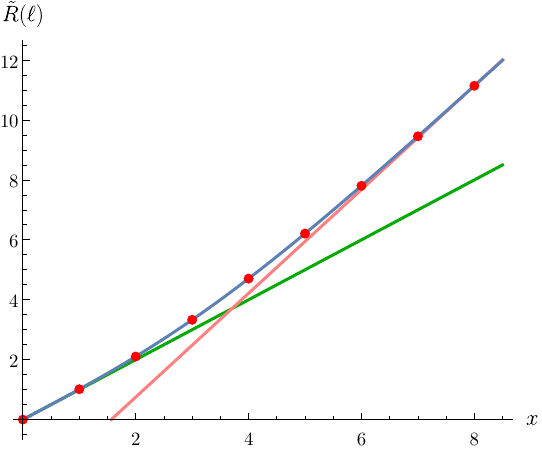}\qquad\qquad
  \includegraphics[height=0.2\textheight]{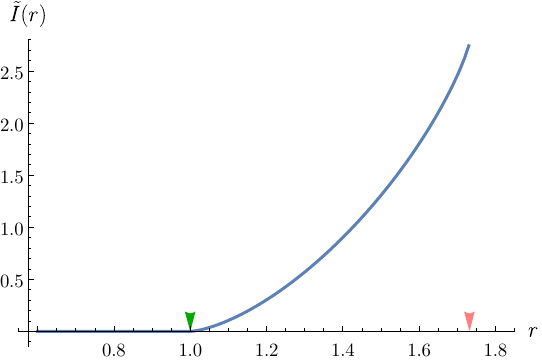}
  \caption{Left: Numerical evaluation of \eqref{formulaWr} (red dots). The blue solid line represents \eqref{AnarateLr}. The straight lines are the asymptotes with slope $1$ (green) and slope $\sqrt3$ (pink).\\
  Right: Radial rate function $\ratefr$, computed analytically from \eqref{AnarateLr} and \eqref{Iradial}. The marked spots on the $r$-axis correspond are the slopes of the asymptotes in the left diagram. }
  \label{fig:rateLr}
\end{figure}
It turns out that the maximum is always attained for $\lambda=\ell/\sqrt3 (1,1,1)$ and $p=0$, and equivalent points. At these points the function is easily evaluated as
\begin{equation}\label{AnarateLr}
  \rateLr(\ell)= 2{\mathrm{arcosh}}\Bigl( \cosh(\ell/2\sqrt3)^3\Bigr).
\end{equation}
The Legendre transform $\ratefr$ can be evaluated analytically (see Figure \ref{fig:rateLr}), but the expression is not very enlightening. We have $\ratefr(r)=0$ for $r<1$, i.e., in the propagation region, and $\ratefr(r)=\infty$ for $r>\sqrt3$, which corresponds to the largest jump vector $(1,1,1)$ (and equivalents) visible in \eqref{bb94spherical} as the appearance of trigonometric power $\exp(i(p_1+p_2+p_3)$. This corresponds to the trivial bound in Lemma~\ref{lem:Iinf}. For intermediate $r$ we have, for example,
$\ratefr(1.1)=.104$, so that the probability for position $Q(t)$ outside a ball of radius $1.1\,t$ after $t$ steps decays at least like $\exp\bigl(-t\ratefr(1.1)\bigr)\approx.9^t$.



\section*{Acknowledgements}
The main idea and body of the paper originate from a visit of A.J. to Hannover in 2012.  We decided to wrap up and complete it, because, to the best of our knowledge, it still provides the best result of its kind. 
The authors thank Andr\'e Ahlbrecht for stimulating discussions on early drafts of this manuscript.
C. Cedzich acknowledges partial support by the Deutsche Forschungsgemeinschaft (DFG, German Research Foundation) -- 441423094.


\bibliographystyle{abbrvArXiv}
\bibliography{extail}

\begin{appendix}
\section{Continuity of the supremum}\label{app:sup_cont}

In this appendix we prove the following topological Lemma. The proof is taken and adapted from \cite{WWong} and, besides completeness, is stated here to prevent the fleeting nature of blog posts:
\begin{lem}
	Let $X,Y$ be topological spaces with $Y$ compact, and let $f:X\times Y\to\Rl$ be (jointly) continuous. Then $g(x):=\sup_{y\in Y}f(x,y)$ is well-defined and continuous.
\end{lem}

\begin{proof}
	We prove the statement in three steps: first, for fixed $x\in X$, $f(x,\cdot):Y\to\Rl$ is continuous (by assumption) and bounded (since $Y$ is compact). Therefore, $g(x)<\infty$. 	
	
	Next, note that for $a,b\in\Rl$ the open sets $(-\infty,a)$ and $(b,\infty)$ form a subbase for the topology on $\Rl$ (the finite intersections of these sets form a basis for the standard topology on $\Rl$, see e.g. \cite[Appendix A2]{rudinFunctionalAnalysis1991}). Thus, the statement of the lemma follows from the openness of $g^{-1}((-\infty,a))$ and $g^{-1}((b,\infty))$.
	
	Let $\pi_X:X\times Y\to X$ be the canonical projection onto the first factor, which is open and continuous by definition. Clearly, $g^{-1}((b,\infty))=(\pi_X\circ f^{-1})((b,\infty))$, so that the continuity of $f$ implies that $g^{-1}((b,\infty))$ is open for any $b\in\Rl$.
	
	To show that also $g^{-1}((-\infty,a))$ is open we use the compactness of $Y$: first, note that $g(x)<a$ implies that $f(x,y)<a$ for all $y\in Y$ (by definition of $g$). But this is just saying that the set
	\begin{equation*}
		\big\{(x,y)\mid g(x)<a,(x,y)\in f^{-1}((-\infty,a))\big\}\subset X\times Y
	\end{equation*}
	is open. This implies that for every $x\in g^{-1}((-\infty,a))$ and every $y\in Y$ there is a neighbourhood $U_{(x,y)}\times V_{(x,y)}$ that is contained in $f^{-1}((-\infty,a))$ (by the definition of openness).
	Since $Y$ is compact, a finite subset $\{(x,y_i)\}_{i=1}^n$ labels such boxes that cover $\{x\}\times Y$, and hence
	\begin{equation*}
		\{x\}\times Y\subset\left[\bigcap_{i=1}^nU_{(x,y_i)}\right]\times Y\subset f^{-1}((-\infty,a)).
	\end{equation*}
	This implies that $g^{-1}((-\infty,a))\equiv \bigcup_{x\in g^{-1}((-\infty,a))}\bigcap_{i=1}^{n(x)}U_{(x,y_i)}$ is open, where $n(x)<\infty$ for all $x\in X$.
\end{proof}
\end{appendix}

\end{document}